\documentclass[a4paper,11pt]{article}

\usepackage[utf8]{inputenc}
\usepackage[margin=1in]{geometry}
\usepackage{graphicx}
\usepackage{subcaption}
\usepackage{amsmath, amssymb, amsthm}
\usepackage{booktabs}
\usepackage{enumerate,paralist}
\usepackage{lineno}
\usepackage{url}
\usepackage[strict]{changepage}
\usepackage[dvipsnames]{xcolor}
\usepackage[nocompress]{cite}
\usepackage{titlesec}
\usepackage[pdfpagelabels,colorlinks,allcolors=blue]{hyperref}
\usepackage{microtype,verbatim,dsfont}

\theoremstyle{plain}

\newtheorem{theorem}{Theorem}
\newtheorem{corollary}{Corollary}

\newtheorem{fact}{Fact}
\newtheorem{proposition}{Proposition}
\newtheorem{lemma}{Lemma}

\newtheorem{pattern}{Forbidden~Pattern}
\theoremstyle{definition}  

\newcommand{\df}[1]{\emph{#1}}
\newcommand{\stel}[2]{s_{#2}^{#1}}

\graphicspath{{pics/}}

\newcommand{\NP}{{\ensuremath{\mathcal{NP}}}}

\newcommand{\BS}{{\color{Blue}\ensuremath{\mathcal{B}lue}}}
\newcommand{\RS}{{\color{Red}\ensuremath{\mathcal{R}ed}}}
\newcommand{\GS}{{\color{Green}\ensuremath{\mathcal{G}reen}}}

\begin{document}

\title{Four Pages Are Indeed Necessary for Planar~Graphs}

\author{%
	Michael~A.~Bekos$^1$,
	Michael~Kaufmann$^1$, 
	Fabian~Klute$^{2,3}$,\\ 
	Sergey~Pupyrev$^4$,
	Chrysanthi~Raftopoulou$^5$,
	Torsten~Ueckerdt$^6$
\\
\medskip
\\
\small$^1$Wilhelm-Schickhard-Institut f\"ur Informatik, Universit\"at T\"ubingen, T\"ubingen, Germany\\
\small\texttt{\{bekos,mk\}@informatik.uni-tuebingen.de}
\\
\small$^2$Institute of Logic and Computation, Technische Universit\"at Wien, Wien, Austria\\
\small$^3$Deptartment of Information and Computing Sciences, Utrecht University, the Netherlands\\
\small\texttt{f.m.klute@uu.nl}
\\
\small$^4$Facebook, Inc., Menlo Park, CA, USA\\
\small\texttt{spupyrev@gmail.com}
\\
\small$^5$School of Applied Mathematical \& Physical Sciences, NTUA, Athens, Greece\\
\small\texttt{crisraft@mail.ntua.gr}
\\
\small$^6$Institute of Theoretical Informatics, Karlsruhe Institute of Technology, Karlsruhe, Germany\\
\small\texttt{torsten.ueckerdt@kit.edu}	
}

\date{}
\maketitle

\begin{abstract}  
An embedding of a graph in a book consists of a linear order of its vertices
along the spine of the book and of an assignment of its edges to the pages of
the book, so that no two edges on the same page cross. The book thickness of a
graph is the minimum number of pages over all its book embeddings. Accordingly,
the book thickness of a class of graphs is the maximum book thickness over all
its members. %
In this paper, we address a long-standing open problem regarding the exact book
thickness of the class of planar graphs, which previously was known to be either
three or four. We settle this problem by demonstrating planar graphs that
require four pages in any of their book embeddings, thus establishing that the
book thickness of the class of planar graphs is four.
\end{abstract}

\section{Introduction}
\label{sec:introduction}

Embedding graphs in books is a fundamental problem in graph theory, which has
been the subject of intense research over the years mainly due to the numerous
applications that it
finds~\cite{DBLP:journals/jgaa/BiedlSWW99,DBLP:conf/compgeom/BinucciLGDMP19,DBLP:journals/jal/Wood02,CLR87,DBLP:conf/focs/Jacobson89,DBLP:journals/siamcomp/MunroR01,DBLP:journals/tc/Rosenberg83,DBLP:conf/stoc/Pratt73,DBLP:journals/jacm/Tarjan72}. Seminal results date back to the 70s by Ollmann~\cite{Oll73}, while several important milestones appear regularly over the years~\cite{DBLP:journals/jct/BernhartK79,DBLP:conf/compgeom/BinucciLGDMP19,CLR87,DBLP:journals/dam/GanleyH01,DBLP:conf/focs/Heath84,DBLP:journals/jcss/Yannakakis89}. In a \df{book embedding} of a graph, the vertices are restricted to a line, called the~\df{spine} of the book, and the edges are assigned to different half-planes delimited by the spine, called \df{pages} of the book, so that no two edges on the same page cross; see Fig.~\ref{fig:example}. The \df{book thickness} (or  \df{stack~number} or \df{page~number}) of a graph is the minimum number of pages required by any of its book~embeddings.

Back in 1979, Bernhart and Kainen preliminary observed that the book thickness
of a graph can be linear in the number of its vertices; for instance, the book
thickness of the complete $n$-vertex graph $K_n$ is $\lceil n/2 \rceil$;
see~\cite{DBLP:journals/jct/BernhartK79}. Sublinear bounds on the book thickness
are known for several classes of graphs;
see~\cite{DBLP:journals/jal/Malitz94,DBLP:journals/jal/Malitz94a,DBLP:journals/dcg/DujmovicW07,DBLP:journals/dam/GanleyH01,Bla03,DBLP:books/daglib/0030491,DBLP:journals/algorithmica/BekosBKR17}. The most notable such class seems to be the one of planar graphs, as is evident from the numerous papers that have been published on the topic over the years~\cite{DBLP:conf/stoc/BussS84,DBLP:conf/focs/Heath84,Istrail1988a,DBLP:journals/jct/BernhartK79,NC08,DBLP:journals/appml/KainenO07,DBLP:journals/mp/CornuejolsNP83,DBLP:conf/cocoon/RengarajanM95,DBLP:journals/dcg/FraysseixMP95,DBLP:journals/algorithmica/BekosGR16,DBLP:conf/esa/0001K19,Ewald1973,GH75,DBLP:conf/focs/Heath84,DBLP:journals/dam/GuanY2019,DBLP:conf/stoc/Yannakakis86,DBLP:journals/jcss/Yannakakis89,DBLP:conf/gd/Pupyrev17,DBLP:conf/wg/AlamBG0P18}. In particular, the graphs with book thickness one are precisely the outerplanar graphs~\cite{DBLP:journals/jct/BernhartK79}. The graphs with book thickness at most two are the subgraphs of planar Hamiltonian graphs~\cite{DBLP:journals/jct/BernhartK79}, which include planar bipartite~\cite{DBLP:journals/dcg/FraysseixMP95} and series-parallel graphs~\cite{DBLP:conf/cocoon/RengarajanM95}.

The study of the book thickness of general planar graphs was initiated by
Leighton, who back in the 80s asked whether their book thickness is bounded by a
constant; see~\cite{DBLP:conf/stoc/BussS84}. The first positive answer to this
question was given by Buss and Shor~\cite{DBLP:conf/stoc/BussS84}, who proposed
a simple recursive (on the number of separating triangles) algorithm to embed
every planar graph in books with nine pages; note that a planar graph without
separating triangles is Hamiltonian~\cite{Wig82}, and thus embeddable in books
with two pages.

The bound of nine pages by Buss and Shor was improved to seven by
Heath~\cite{DBLP:conf/focs/Heath84}, who introduced an important methodological
foundation called \df{peeling-into-levels}\footnote{In the literature, sometimes
	this technique is erroneously attributed to
	Yannakakis~\cite{DBLP:journals/jcss/Yannakakis89}.}, according to which the
vertices of a planar graph are partitioned into levels such that (i) the
vertices on its unbounded face are at level $0$, and (ii) the vertices that are
on the unbounded face of the subgraph induced by deleting all vertices of levels
$ \leq i-1$ are at level $i$ ($0 < i < n$). It is not difficult to see that each
connected component of the subgraphs induced by the vertices of the same level
is an outerplanar graph, and thus embeddable in a single
page~\cite{DBLP:journals/jct/BernhartK79}. Hence, the main challenge is to embed
the remaining edges, that is, those connecting vertices in consecutive~levels.

Heath~\cite{DBLP:conf/focs/Heath84} managed to address this challenge with a
relatively simple algorithm that uses six pages. In a subsequent work, which is
probably the most cited in the field,
Yannakakis~\cite{DBLP:journals/jcss/Yannakakis89} improved upon Heath's
algorithm. Using the peeling-into-levels technique, he proposed a simple
algorithm that yields embeddings in books with five pages (even though, the
details of the algorithm are left to the reader). With a more complicated and
involved algorithm, which is based on distinguishing different cases of the
underlying order and the edges to be embedded, Yannakakis reduced the required
number of pages to four, which is currently the best-known upper bound on the
book thickness of planar graphs.

The currently best-known lower bound is usually attributed to Goldner and
Harary~\cite{GH75}, who proposed the smallest maximal planar graph that is not
Hamiltonian, and therefore not embeddable in books with two pages; see
Fig.~\ref{fig:goldner-harary-1}. However, this particular graph is a planar
$3$-tree and by a result of Heath~\cite{DBLP:conf/focs/Heath84}, it is
embeddable in a book with three pages; see Fig.~\ref{fig:goldner-harary-2}. Note
that determining the exact book thickness of a planar graph turns out to be an
\NP-complete problem, even for maximal planar graphs~\cite{Wig82}.

\begin{figure}[t]
	\centering
	\begin{subfigure}[b]{.4\textwidth}
	\centering
	\includegraphics[scale=0.9,page=1]{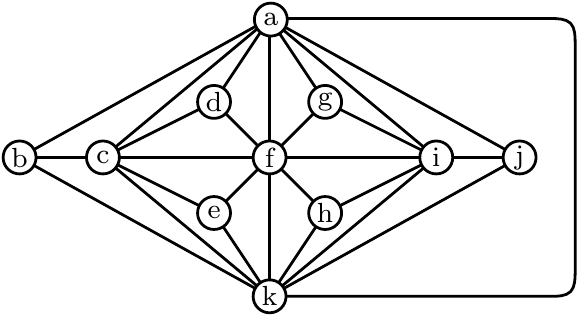}
	\caption{}
	\label{fig:goldner-harary-1}
	\end{subfigure}
	\begin{subfigure}[b]{.59\textwidth}
	\centering
	\includegraphics[scale=1,page=3]{goldner-harary}
	\caption{}
	\label{fig:goldner-harary-2}
	\end{subfigure}
	\caption{%
	Illustration of 
	(a)~the Goldner-Harary graph and	
	(b)~its $3$-page book embedding in which
	edges assigned to different pages are colored differently.}
	\label{fig:example}
\end{figure}

To the best of our knowledge, there is no planar graph described in the
literature that requires more than three pages despite various efforts. In an
extended abstract of~\cite{DBLP:journals/jcss/Yannakakis89}, which appeared at
STOC in 1986~\cite{DBLP:conf/stoc/Yannakakis86}, Yannakakis claimed the
existence of such a graph and provided a sketch of a proof; notably the
arguments in this sketch seem to be sound apart from the fact that some of the
gadget-graphs that are central in the proof are not defined. The details of this
sketch, however, never appeared in a paper. Furthermore, the proof-sketch was
not part of the subsequent journal
version~\cite{DBLP:journals/jcss/Yannakakis89} of the extended
abstract~\cite{DBLP:conf/stoc/Yannakakis86}. Thus the problem of determining
whether there exists a planar graph that requires four pages still remains
unsolved, as also noted by Dujmovi{\'{c}} and
Wood~\cite{DBLP:journals/dcg/DujmovicW07} in 2007, and clearly forms the most
intriguing open problem in the field. Note that, in the same work,
Dujmovi{\'{c}} and Wood proposed a planar graph that might require four pages in
any of its book embeddings. However, they had overlooked a previous result by
Heath~\cite{DBLP:conf/focs/Heath84} regarding the book thickness of planar
$3$-trees, which immediately implies that their claim was not valid. %
A more recent attempt to find a planar graph that requires four pages in any of
its book embeddings was made by Bekos, Kaufmann, and
Zielke~\cite{DBLP:conf/gd/Bekos0Z15}, who proposed a formulation of the problem
of testing whether a given (not necessarily planar) graph admits an embedding
into a book with a certain number of pages as a SAT instance, and systematically
tested several hundred maximal planar graphs but without any particular~success.
Later Pupyrev~\cite{DBLP:conf/gd/Pupyrev17} computed book embeddings of \df{all}
maximal planar graphs of size $n \le 18$ and found no instance that requires
four pages.

\medskip\noindent\textbf{Our contribution.} In this paper, we address  the aforementioned long-standing open problem. 
Our main result is summarized in the following theorem. 

\begin{theorem}\label{thm:main}
There exist planar graphs that do not admit $3$-page book embeddings.
\end{theorem}

\noindent Together with Yannakakis' upper bound of four~\cite{DBLP:journals/jcss/Yannakakis89}, Theorem~\ref{thm:main} implies the following corollary.

\begin{corollary}\label{cor:main}
The book thickness of the class of planar graphs is four.
\end{corollary}

\noindent We provide two proofs of Theorem~\ref{thm:main}. The first one is
combinatorial (with some computer-aided prerequisites) and regards a
significantly large planar graph. After recalling basic notions and results on
book embeddings in Section~\ref{sec:preliminaries}, we describe the construction
of this graph in Section~\ref{sec:graph}, where we also present two properties
of a particular subgraph of it, which have been verified by a computer (refer to
Facts~\ref{fact:1} and~\ref{fact:2}). In Section~\ref{sec:combinatorial}, we
prove that the graph presented in Section~\ref{sec:graph} does not admit a
$3$-page book embedding. We give the main ingredients of this proof in
Section~\ref{ssec:idea}, while in Section~\ref{ssec:cases} we investigate a
systematic analysis of cases of different underlying linear orders to conclude
our main result.

The second proof of Theorem~\ref{thm:main} is purely computer-aided; see
Section~\ref{sec:sat}. With two independent implementations~\cite{alice,bob} of
the SAT formulation presented in~\cite{DBLP:conf/gd/Bekos0Z15}, we confirm that
a particular maximal planar graph with $275$ vertices does not admit a $3$-page
book embedding; see Fig.~\ref{fig:needs4pages} for an illustration of the graph.
Key in our approach is the introduction of several symmetry-breaking constraints
in the SAT instance. These constraints helped in reducing the search space of
possible satisfying assignments and made the instance verifiable using modern
SAT solvers. We conclude in Section~\ref{sec:conclusions} with several open
problems.

\section{Preliminaries}
\label{sec:preliminaries}

A \df{vertex ordering} $\prec$ of a simple undirected graph $G=(V, E)$ is a
total order of its vertex set $V$, such that for any two vertices $u$ and $v$,
$u \prec v$ if and only if $u$ precedes $v$ in the order. Two vertices $u$ and
$v$ are said to be \df{on opposite sides} of an edge $(x,y)$, where $u \prec v$
and $x \prec y$, if $u \prec x \prec v \prec y$ or $x \prec u \prec y \prec v$.
Otherwise, $u$ and $v$ are \df{on the same side} of $(x,y)$. We write $[v_1,
v_2, \ldots, v_k ]$ to denote $v_i \prec v_{i+1}$ for all $1 \leq i < k$. Let 
$F$ be a set of $k \geq 2$ independent pairs of vertices $\langle s_i, t_i
\rangle$, that is, $F=\{\langle s_i,t_i \rangle;\;i=1,2,\ldots,k\}$. Assume
without loss of generality that $s_i \prec t_i$, for all $1 \leq i \leq k$. If
the order is $[s_1, \ldots, s_k, t_k, \ldots, t_1]$, then we say that the pairs
of $F$ form a \df{$k$-rainbow}, while if the order is $[s_1, t_1, \ldots, s_k,
t_k]$, then the pairs of $F$ form a \df{$k$-necklace}. The pairs of $F$ form a
\df{$k$-twist} if the order is $[s_1, \ldots, s_k, t_1, \ldots, s_k]$. Note that
since each edge is defined by a pair of vertices, the three definitions are
directly extendable to independent edges; see Fig.~\ref{fig:necklace-twist}. For
this case, two independent edges that form a $2$-twist (respectively,
$2$-rainbow, $2$-necklace) are commonly referred to as \df{crossing}
(respectively, \df{nested}, \df{disjoint}).

A \df{$k$-page book embedding} of a graph is a pair $\mathcal{E} = (\prec,
\{E_1, \dots, E_k\})$, where $\prec$ is a vertex ordering of $G$ and $\{E_1,
\dots, E_k\}$ is a partition of $E$ into \df{pages}, that is, sets of pairwise
non-crossing edges. Equivalently, a $k$-page book embedding is a vertex ordering
and a $k$-edge-coloring such that no two edges of the same color cross with
respect to the ordering. The \df{book thickness} of a graph $G$ is the minimum
$k$ such that $G$ admits a $k$-page book embedding. As noted in several papers,
a $k$-page book embedding, $\mathcal{E}$, can be transformed into a circular
embedding, $C(\mathcal{E})$, with a $k$-edge-coloring in which all vertices
appear in the circumference of a circle in the same order as in $\mathcal{E}$
and the edges are drawn as straight-line segments in the interior of the circle,
such that no two edges of the same color cross, and vice
versa~\cite{DBLP:journals/jct/BernhartK79,DBLP:conf/focs/Heath84}; see
Fig.~\ref{fig:AxxByy}. The next lemma, whose proof is immediate, provides
sufficient conditions for the non-existence of a $3$-page book embedding.

\begin{figure}[t!]
	\centering	
	\begin{subfigure}[b]{.32\textwidth}
		\centering
		\includegraphics[scale=1,page=1]{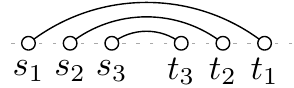}
		\caption{}
		\label{fig:rainbow}
	\end{subfigure}
	\hfil
	\begin{subfigure}[b]{.32\textwidth}
		\centering
		\includegraphics[scale=1,page=2]{preliminaries}
		\caption{}
		\label{fig:twist}
	\end{subfigure}
	\hfil
	\begin{subfigure}[b]{.32\textwidth}
		\centering
		\includegraphics[scale=1,page=3]{preliminaries}
		\caption{}
		\label{fig:necklace}
	\end{subfigure}
	\caption{Illustration of three edges that form:
		(a)~a $3$-rainbow, 
		(b)~a $3$-twist, and 
		(c)~a $3$-necklace.}
	\label{fig:necklace-twist}
\end{figure}

\begin{lemma}\label{lem:no3}
A $3$-page book embedding of a graph does not contain:
\begin{inparaenum}[(i)]
\item \label{no3:twist} four edges that form a $4$-twist in~$\prec$, 
\item \label{no3:cross2} a pair of crossing edges that both cross two edges assigned to two different pages, and
\item \label{no3:cross3} an edge crossing three edges assigned to three different pages.
\end{inparaenum}
\end{lemma}

\noindent The next result by Erd\H{o}s and Szekeres~\cite{CM_1935__2__463_0} is used to simplify our case analysis.

\begin{lemma}[Erd\H{o}s and Szekeres~\cite{CM_1935__2__463_0}]\label{lem:es}
Given $a, b \in \mathbb{N}$, every sequence of distinct real numbers of length at least $a\cdot b + 1$ 
contains a monotonically increasing subsequence of length $a+1$ or 
a monotonically decreasing subsequence of length $b+1$.
\end{lemma}

Lemma~\ref{lem:es} implies that, for every $r \ge 1$, if the input graph has
sufficiently many independent edges, then one can always find $r$ of them that
form an $r$-rainbow or an $r$-twist or an $r$-necklace in every ordering
$\prec$. To see this, assume that the graph contains $r^3$ independent edges.
Represent each edge connecting the $i$-th with the $j$-th vertex in $\prec$ by a
pair $(i, j)$ with $i < j$. Consider the pairs sorted by the first coordinates,
and apply Lemma~\ref{lem:es} with $a=r^2$ and $b=r-1$ to the second coordinates
of the edges. Then, either
\begin{inparaenum}[(i)]
\item \label{c:inc} there exists $r^2+1$ edges such that every pair of them forms a $2$-twist or a $2$-necklace (corresponding to an increasing subsequence), which implies that $r$ of them  form an $r$-twist or an $r$-necklace~\cite{DBLP:conf/gd/AlamBG0P18}, or
\item \label{c:dec} there exists an $r$-rainbow (corresponding to a decreasing subsequence).
\end{inparaenum}
Note that the same argument can be applied to $r^3$ designated pairs of vertices (not necessarily connected by an edge); thus we have the following corollary.

\begin{corollary}\label{lem:pair-conf}
	For every vertex ordering, $\prec$, of a graph with $r^3$ designated pairs of vertices, one can identify $r$ pairs 
that form either an $r$-rainbow or an $r$-twist or an $r$-necklace in $\prec$.
\end{corollary}

\section{The Basic Graph Structure}
\label{sec:graph}

The graph used to prove Theorem~\ref{thm:main} is built using a sequence of
gadgets---planar graphs that do not admit a $3$-page book embedding under
certain conditions. To define a gadget, denoted by $Q_k$, we recall the
operation of the \df{stellation of a face $f$}, that is, the addition of a
vertex in the interior of $f$ connected to all vertices delimiting $f$.
Accordingly, the operation of stellating a plane graph consists of stellating
all its bounded faces.

For $k \geq 2$, graph $Q_k$ is a plane graph, which contains as a subgraph the
complete bipartite graph $K_{2,k}$ with bipartitions $\{A,B\}$ and
$\{t_0,\ldots,t_{k-1}\}$; see Fig.~\ref{fig:K2k}. We choose the embedding of
$Q_k$ such that the faces of $K_{2,k}$ are $F_i=\langle A, t_i, B, t_{i+1}
\rangle$ for $i=0,\ldots,k-1$ (indices taken modulo $k$) with $F_{k-1}$ being
its outerface. We refer to vertices $A$ and $B$ as the \df{poles} of $Q_k$, and
to the vertices $t_0,\ldots,t_{k-1}$ as the \df{terminals} of $Q_k$. For $i =
0,\ldots,k-2$, we call terminals $t_i$ and $t_{i+1}$ of $Q_k$ \df{consecutive};
notice that $t_0$ and $t_{k-1}$ are not consecutive by the definition.

\begin{figure}[t]
	\centering
	\begin{subfigure}[b]{.48\textwidth}
	\centering
	\includegraphics[scale=.9,page=2]{pics/Q}
	\caption{}
	\label{fig:K2k}
	\end{subfigure}	
	\hfil
	\begin{subfigure}[b]{.48\textwidth}
	\centering
	\includegraphics[scale=.9,page=1]{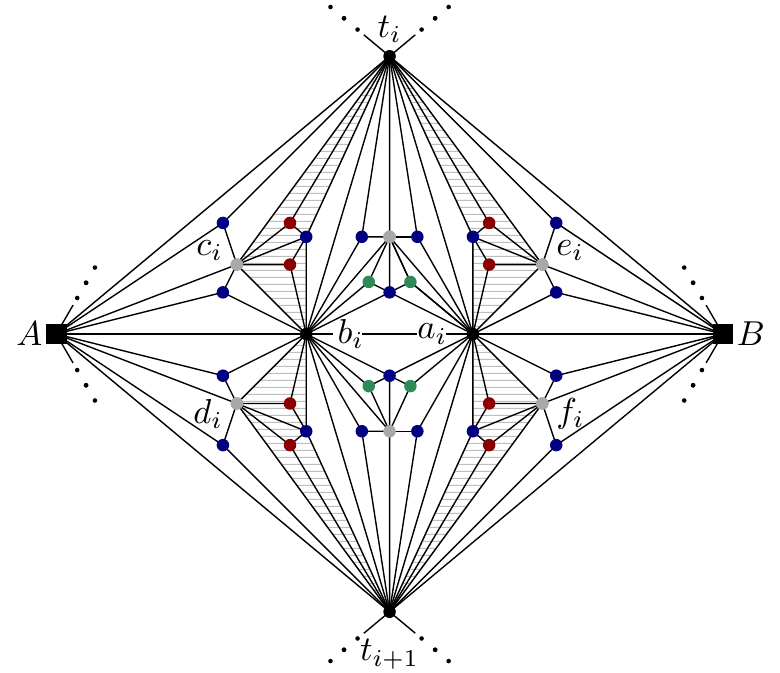}
	\caption{}
	\label{fig:qi}
	\end{subfigure}
	\caption{Illustration for the construction of graph $Q_k$.}
	\label{fig:qk}
\end{figure}

Let $i \in \{0,\ldots,k-2\}$. In $Q_k$, vertices $A$ and $B$ are connected by a
path of length~3 which is embedded within $F_i$ and consists of the following
three edges: $(A,b_i)$, $(b_i,a_i)$ and $(a_i,B)$; see Fig.~\ref{fig:qi}. We
refer to the two vertices $a_i$ and $b_i$ of this path as the \df{satellites} of
the (consecutive) terminals $t_i$ and $t_{i+1}$; accordingly, we refer to the
edge connecting $a_i$ and $b_i$ as the \df{satellite edge} of $t_i$ and
$t_{i+1}$. Observe that we do not embed any path in $F_{k-1}$. The two faces on
the opposite sides of the path embedded in $F_i$ are triangulated by the edges
$(t_i,a_i), (t_i,b_i)$ as well as $(t_{i+1},a_i)$ and $(t_{i+1},b_i)$. We
proceed by stellating the graph constructed so far twice (refer to the gray and
blue vertices in Fig.~\ref{fig:qi}, respectively). 
Let $c_i$, $d_i$, $e_i$ and $f_i$ be the vertices that stellated $\langle A,
t_i, b_i \rangle$, $\langle A, b_i, t_{i+1} \rangle$, $\langle B, t_i, a_i
\rangle$ and $\langle B, a_i, t_{i+1} \rangle$ in the first round of stellation.
Let $c_i'$, $d_i'$, $e_i'$ and $f_i'$ be the vertices that stellated $\langle
c_i, t_i, b_i \rangle$, $\langle d_i, b_i, t_{i+1} \rangle$, $\langle e_i, t_i,
a_i \rangle$ and $\langle f_i, a_i, t_{i+1} \rangle$ in the second round of
stellation; refer to the blue-colored vertices that lie within the gray-shaded
regions of Fig.~\ref{fig:qi}. We proceed by stellating faces $\langle c_i, c_i',
t_i \rangle$, $\langle c_i, c'_i, b_i \rangle$, $\langle d_i, d_i', b_i
\rangle$, $\langle d_i, d_i', t_{i+1} \rangle$, $\langle e_i, e_i', t_i
\rangle$, $\langle e_i, e_i', a_i \rangle$, $\langle f_i, f_i', a_i \rangle$ and
$\langle f_i, f_i', t_{i+1} \rangle$; refer to the red-colored vertices of
Fig.~\ref{fig:qi}. The satellite edge $(a_i,b_i)$ delimits two faces, each of
which is neighboring two other faces that we stellate; refer to the
green-colored vertices of Fig.~\ref{fig:qi}. Edge $(A,B)$ completes the
construction of $Q_k$. 
Note that graph $Q_2$, the first member in the described family, consists of
$42$ vertices and $126$~edges.

The following two facts that hold for certain members of the constructed graph
family have been verified by a computer using the SAT-formulation proposed
in~\cite{DBLP:conf/gd/Bekos0Z15}; we provide further details in
Section~\ref{sec:sat}. We use these facts in the combinatorial proof of
Theorem~\ref{thm:main}.

\begin{fact}\label{fact:1}
Graph $Q_k$ with $k \geq 7$ does not admit an embedding in a book with three pages, 
$\BS$, $\RS$ and $\GS$, under the following restrictions:
\begin{inparaenum}[(i)]
\item \label{fact1:i} the poles $A$ and $B$ are consecutive in the ordering, 
\item \label{fact1:ii} all edges from $A$ to the terminals of $Q_k$ belong to $\BS$, and
\item \label{fact1:iii} all edges from $B$ to the terminals of $Q_k$ belong to $\RS$ or $\GS$.
\end{inparaenum}
\end{fact}	

\begin{fact}\label{fact:2}
Graph $Q_k$ with $k \geq 10$ does not admit an embedding in a book with three pages, 
$\BS$, $\RS$ and $\GS$, under the following restrictions:
\begin{inparaenum}[(i)]
\item \label{fact2:i} all terminals of $Q_k$ are on the same side of $(A,B)$,
\item \label{fact2:ii} all edges from $A$ to the terminals of $Q_k$ belong to $\BS$, and
\item \label{fact2:iii} all edges from $B$ to the terminals of $Q_k$ belong to $\RS$.
\end{inparaenum}
\end{fact}

\noindent Note that Fact~\ref{fact:1} imposes stronger restrictions in the
vertex ordering than Fact~\ref{fact:2}, while Fact~\ref{fact:2} imposes stronger
restrictions to the edges adjacent to $A$ and $B$. In the remainder, we denote
by $Q$ the smallest member of the constructed family of graphs for which both
Facts~\ref{fact:1} and~\ref{fact:2} hold:
\[
Q:=Q_{10}
\]

Consider a plane graph $G$ and let $H$ be a plane graph with two designated
vertices $A$ and $B$ that appear consecutively along its outerface. The
operation of \df{attaching} $H$ along an edge $(u,v)$ of $G$ consists of
removing $(u,v)$ from $G$ and of introducing $H$ into $G$ by identifying vertex
$A$ of $H$ with vertex $u$ of $G$ and vertex $B$ of $H$ with vertex $v$ of $G$;
see Fig.~\ref{fig:attachments}. The obtained graph is clearly planar, since both
$G$ and $H$ are planar and simple due to removal of $(u,v)$ from $G$.

\begin{figure}[t]
	\centering
	\includegraphics[scale=1,page=3]{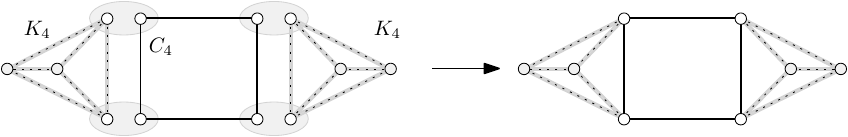}
	\caption{Attaching two copies of the complete graph $K_4$ along two edges of a $4$-cycle $C_4$.}
	\label{fig:attachments}
\end{figure}

\section{A Combinatorial Proof with Computer-Aided Prerequisites}
\label{sec:combinatorial}

In this section, we construct a planar graph $G$ containing several copies of
$Q$. Using as prerequisites Facts~\ref{fact:1} and~\ref{fact:2}, we explore
certain properties of graph $G$ (Section~\ref{ssec:idea}) to prove that it does
not admit a $3$-page book embedding by analyzing possible vertex orderings
(Section~\ref{ssec:cases}).

\subsection{The Idea}
\label{ssec:idea}

We prove Theorem~\ref{thm:main} by contradiction, that is, by assuming that $G$
admits a book embedding $\mathcal E$ with three pages denoted by $\BS$, $\RS$,
and $\GS$. Graph $G$ contains as a subgraph a \df{base graph}, which we denote
by $G_N$, consisting of a large number $N \gg 1$ of copies of graph $Q$ that
share the same pair of poles, $A$ and $B$, and edge $(A, B)$. Hence, graph $G_N$
is symmetric with respect to $A$ and $B$. %
Let $n_Q$ and $m_Q$ be the number of vertices and edges in $Q$, and let $b_Q$ be
the number of $3$-page book embeddings of graph $Q$. Clearly $b_Q$ is upper
bounded by $3^{m_Q} \cdot n_Q!$; it follows that if $N$ is at least $\kappa
\cdot 3^{m_Q} \cdot n_Q!$, then by pigeonhole principle $G_N$ contains $\kappa$
copies of graph $Q$ with the \df{majority property}, that is, corresponding
vertices of $Q$ in each of these $\kappa$ copies appear in the same relative
order in $\mathcal{E}$, and additionally the edges that connect these vertices
in each of the copies are assigned to the same pages. We refer to two vertices
that correspond to the same vertex in $Q$ and that belong to different copies
satisfying the majority property as \df{twin vertices}. Accordingly, two edges
connecting twin vertices are called \df{twin edges}.

\begin{lemma}
	\label{lm:twin-edges}
	A pair of independent twin edges either form a $2$-rainbow or a $2$-necklace in $\mathcal{E}$.
\end{lemma}
\begin{proof}
	Observe that two independent twin edges cannot form a $2$-twist, as
	they are assigned to the same page in $\mathcal{E}$ by the majority property.
\end{proof}

Next we further increase $N$ to guarantee an additional property, called the
\df{monotonic property}, for the $\kappa$ copies of graph $Q$ that comply with
the majority property. Denote  by $p_Q$ the number of pairs of vertices in $Q$,
that is, $p_Q = \frac{n_Q(n_Q-1)}{2}$. By Corollary~\ref{lem:pair-conf}, if $N$
is at least $\kappa^{3 \cdot p_Q} \cdot 3^{m_Q} \cdot n_Q!$,
then one can identify $\kappa$ copies of $Q$ in $G_N$ complying with the
majority property, such that, for each pair of vertices of $Q$, the
corresponding pairs of vertices in these $\kappa$ copies form a $\kappa$-rainbow
or $\kappa$-twist or a $\kappa$-necklace in $\mathcal{E}$. We specify $\kappa$
in the case analysis of Section~\ref{ssec:cases}.

While we mainly focus on the base graph $G_N$, to facilitate our analysis in
cases in Section~\ref{ssec:cases}, we perform an augmentation step that
completes the construction of $G$. Let $H_N$ be a copy of the base graph $G_N$.
We attach a copy of $H_N$ along every satellite edge of the base graph $G_N$. We
refer to the obtained graph as the \df{final graph} $G$, which by construction
is biconnected; the poles of the base graph $G_N$ and the endvertices of each of
its satellite edges are separation pairs in $G$. Next we investigate all
possible vertex orderings of $G$ in its $3$-page book embedding $\mathcal{E}$.

\subsection{Case analysis}
\label{ssec:cases}

Consider the base graph $G_N$ and let $\mathcal{Q}_1,\ldots,\mathcal{Q}_\kappa$
be the $\kappa$ copies of graph $Q$ that comply with the majority and the
monotonic properties. Assuming that $A$ is the first vertex in $\mathcal{E}$, we
consider two main cases in our proof:

\medskip
\begin{compactenum}[\bf {C.}1.]
\item \label{c:AxBy} There exist two terminals of $\mathcal{Q}_1$ that are on opposite sides of edge $(A,B)$ in $\mathcal{E}$.
\item \label{c:AxyB} All terminals of $\mathcal{Q}_1$ are on the same side of $(A,B)$ in $\mathcal{E}$. 
\end{compactenum}

\medskip\noindent\textbf{Case~C.\ref*{c:AxBy}}: We first rule out
Case~C.\ref{c:AxBy} in which there exist two terminals of $\mathcal{Q}_1$, say
$\langle x_1, y_1 \rangle$, that are on opposite sides of edge $(A,B)$ in
$\mathcal{E}$. Observe that in this case it is not a loss of generality to
assume that $x_1$ and $y_1$ are consecutive in the sequence of terminals of
$\mathcal{Q}_1$. By the majority property, the corresponding terminals $\langle
x_2, y_2 \rangle, \ldots, \langle x_\kappa, y_\kappa \rangle$ of
$\mathcal{Q}_2,\ldots,\mathcal{Q}_\kappa$ are also on opposite sides of edge
$(A,B)$. Let $\langle a_1, b_1 \rangle, \ldots, \langle a_\kappa, b_\kappa 
\rangle$ be the corresponding satellite vertices of $\langle x_1, y_1 \rangle,
\ldots, \langle x_\kappa, y_\kappa  \rangle$. W.l.o.g., we further assume that
$x_1 \prec \ldots \prec x_\kappa$, which by the monotonic property implies that
either $y_1 \prec \ldots \prec y_\kappa$ or $y_\kappa \prec \ldots \prec y_1$.
Since $G_N$ is symmetric with respect to $A$ and $B$, we may further assume that
the ordering of the vertices in $\mathcal{E}$ is either $[A \ldots x_1 \ldots
x_\kappa \ldots B  \ldots y_1 \ldots y_\kappa ]$ or $[A \ldots x_1 \ldots
x_\kappa \ldots B  \ldots y_\kappa \ldots y_1 ]$. We next prove that both
patterns are forbidden, assuming $\kappa=3$. Since $G_N$ is symmetric with
respect to $A$ and $B$, by the majority property we may further assume w.l.o.g.\
that $a_i$ and $x_i$ are on the same side of $(A,B)$, namely, $A \prec a_i \prec
B$ holds, for each $i=1,\ldots,\kappa$.

\begin{pattern}
$[A \dots x_1 \dots x_2 \dots x_3 \dots B \dots y_1 \dots y_2 \dots y_3 \dots]$
\end{pattern}
\begin{proof}
By the monotonic property, it follows that either $a_1 \prec a_2 \prec a_3$ or 
$a_3 \prec a_2 \prec a_1$ holds, and that $b_1 \prec b_2 \prec b_3$ or 
$b_3 \prec b_2 \prec b_1$  holds. We start with a few auxiliary propositions.

\begin{adjustwidth}{1.5em}{}
\begin{proposition}\label{prop:AxBy1}
$A \prec x_3 \prec a_3 \prec a_2 \prec a_1 \prec B$.
\end{proposition}
\begin{proof}
If $a_1 \prec a_2 \prec a_3$, then the twin edges $(a_1, y_1)$, $(a_2, y_2)$ and
$(a_3, y_3)$ form a $3$-twist in $\mathcal{E}$, which contradicts
Lemma~\ref{lm:twin-edges}. Hence, $a_3 \prec a_2 \prec a_1$ must hold. Assume
now that $a_1 \prec x_1$, which by the majority property implies that $a_i \prec
x_i$, for each $i=1,2,3$. Since $a_3 \prec a_2 \prec a_1$ holds, it follows that
the relative order is $[A \dots a_3 \dots a_2 \dots a_1 \dots x_1 \dots x_2
\dots x_3 \dots B]$. Hence, edges $(a_1, y_1)$, $(a_2, B)$, $(a_3, x_3)$ and
$(A, x_2)$ form a $4$-twist in $\mathcal{E}$, which is a contradiction by
Lemma~\ref{lem:no3}.\ref{no3:twist}. Thus, $x_1 \prec a_1$ must hold, which by
the majority property implies that $x_i \prec a_i$, for each $i=1,2,3$. Since
$x_3 \prec a_3$ and $a_3 \prec a_2 \prec a_1$ holds, the proposition follows.
\end{proof}

\noindent Similarly, we can prove the following.\medskip

\begin{proposition}\label{prop:AxBy3}
If $A \prec b_1 \prec B$, then $A \prec b_3 \prec b_2 \prec b_1 \prec x_1 \prec B$.
\end{proposition}
\begin{proposition}\label{prop:AxBy4}
If $B \prec b_1$, then $B \prec b_3 \prec b_2 \prec b_1 \prec y_1$.
\end{proposition}
\end{adjustwidth}

\medskip\noindent Let $i \in \{1,2,3\}$. We consider two cases, depending on
whether $a_i$ and $b_i$ are on the same or on different sides of $(A,B)$. Assume
first the former case. Since $A \prec a_i \prec B$, it follows that \mbox{$A
	\prec b_i \prec B$}. By Propositions~\ref{prop:AxBy1} and~\ref{prop:AxBy3}, the
relative order is $[A \dots b_3 \dots b_2 \dots b_1 \dots x_1 \dots x_3 \dots
a_3 \dots a_2 \dots a_1]$. Hence, edges $(a_1,b_1)$, $(a_2,b_2)$, $(a_3,b_3)$
and $(A,x_1)$ form a $4$-twist; a contradiction by
Lemma~\ref{lem:no3}.\ref{no3:twist}. Assume $a_i$ and $b_i$ are on different
sides of $(A,B)$. By the majority property, $B \prec y_i$. By
Propositions~\ref{prop:AxBy1} and~\ref{prop:AxBy4}, the relative order is $[A
\dots a_3 \dots a_2 \dots a_1 \dots B \dots b_3 \dots b_2 \dots b_1]$, which
implies that edges $(a_1,b_1)$, $(a_2,b_2)$, $(a_3,b_3)$ and $(A,B)$  form a
$4$-twist; a contradiction by Lemma~\ref{lem:no3}.\ref{no3:twist}.
\end{proof}

\newcounter{AxxByy}
\begin{pattern}\label{p:AxxByy}
$[A \dots x_1 \dots x_2 \dots x_3 \dots B \dots y_3 \dots y_2 \dots y_1]$
\end{pattern}
\begin{proof}
Let $i \in \{1,2,3\}$. By the monotonic property, either $a_1 \prec a_2 \prec
a_3$ or $a_3 \prec a_2 \prec a_1$ holds, and either $b_1 \prec b_2 \prec b_3$ or
$b_3 \prec b_2 \prec b_1$  holds. Since $A \prec a_i \prec B$, it follows that
if $a_3 \prec a_2 \prec a_1$, then the twin edges $(a_1,y_1)$, $(a_2,y_2)$ and
$(a_3,y_3)$ form a $3$-twist, which a contradiction by
Lemma~\ref{lm:twin-edges}. Hence, $A \prec a_1 \prec a_2 \prec a_3 \prec B$
holds.

\medskip\noindent We proceed by distinguishing two subcases depending on whether
the satellite vertices $a_i$ and $b_i$ are on the same or different sides of
$(A,B)$. We first consider the former case. Since $A \prec a_1 \prec a_2 \prec
a_3 \prec B$ holds, it follows that either $A \prec b_1 \prec b_2 \prec b_3
\prec B$ or $A \prec b_3 \prec b_2 \prec b_1 \prec B$ holds. If $b_3 \prec b_2
\prec b_1$, then the twin edges  $(b_1,y_1)$, $(b_2,y_2)$ and $(b_3,y_3)$ form a
$3$-twist, which is a contradiction by Lemma~\ref{lm:twin-edges}. Hence, $A
\prec b_1 \prec b_2 \prec b_3 \prec B$ must hold. By the monotonic property, the
partial order of vertices $A$, $B$ and of the vertices in $\{x_i,y_i,a_i,b_i;\
i=1,2,3\}$ is one of the following
FP\ref{p:AxxByy}.\ref{p:AaxbBy}-FP\ref{p:AxxByy}.\ref{p:AbaxBy}; note that the
cases that  corresponds to FP\ref*{p:AxxByy}.\ref{p:AabxBy} and
FP\ref*{p:AxxByy}.\ref{p:AbaxBy} in which the terminal $x_i$ precedes the
satellite vertices $a_i$ and $b_i$, are symmetric to
FP\ref*{p:AxxByy}.\ref{p:AbaxBy} and FP\ref*{p:AxxByy}.\ref{p:AabxBy},
respectively, due to the symmetry of $G_N$ with respect to $A$ and $B$.

\begin{enumerate}[\bf FP\ref*{p:AxxByy}.1]
\item \label{p:AaxbBy} $[A \dots a_1 \dots x_1 \dots b_1 \dots a_2 \dots x_2 \dots b_2 \dots a_3 \dots x_3 \dots b_3 \dots B \dots y_3 \dots y_2 \dots y_1]$

\begin{figure}[t!]
    \centering
    \begin{subfigure}[b]{.24\textwidth}
		\centering
		\includegraphics[width=\textwidth,page=1]{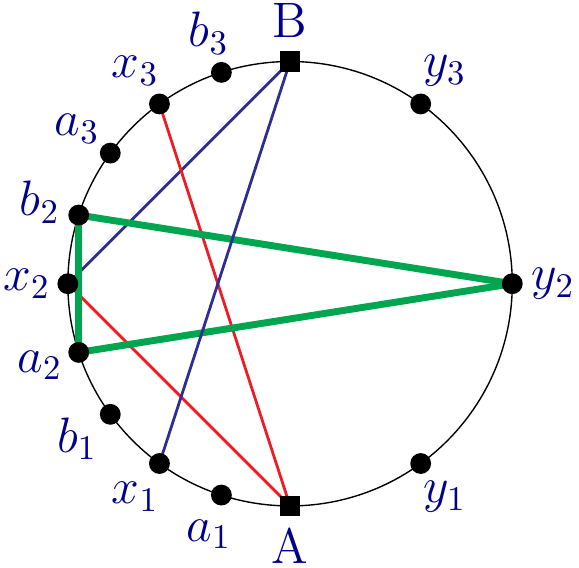}
		\caption{}
		\label{fig:AaxbBy}
	\end{subfigure}
	\hfil
	\begin{subfigure}[b]{.24\textwidth}
		\centering
		\includegraphics[width=\textwidth,page=2]{pics/AxBy}
		\caption{}
		\label{fig:AbxaBy}
	\end{subfigure}
	\hfil
	\begin{subfigure}[b]{.24\textwidth}
		\centering
		\includegraphics[width=\textwidth,page=3]{pics/AxBy}
		\caption{}
		\label{fig:AabxBy}
	\end{subfigure}
	\hfil
	\begin{subfigure}[b]{.24\textwidth}
		\centering
		\includegraphics[width=\textwidth,page=4]{pics/AxBy}
		\caption{}
		\label{fig:AbaxBy}
	\end{subfigure}
    \caption{Illustrations for 
    (a)~FP\ref*{p:AxxByy}.\ref{p:AaxbBy},
    (b)~FP\ref*{p:AxxByy}.\ref{p:AbxaBy},
    (c)~FP\ref*{p:AxxByy}.\ref{p:AabxBy}, and
    (d)~FP\ref*{p:AxxByy}.\ref{p:AbaxBy}.}
    \label{fig:AxxByy}
\end{figure}

Refer to Fig.~\ref{fig:AaxbBy}. Since edges $(A,x_3)$, $(x_1,B)$ and $(a_2,y_2)$
form a $3$-twist, they are assigned to different pages in $\mathcal{E}$. By the
majority property, we may assume that $(A,x_i) \in \RS$, $(a_i,y_i) \in \GS$ and
$(x_i,B) \in \BS$. It follows that $(b_i,y_i) \in \GS$ and $(a_i,b_i) \in \GS$.
Consider now vertex $\stel{aby}{2}$ of $G_N$ that was introduced due to the
stellation of face $\langle a_2, b_2, y_2 \rangle$ in $G_N$. Due to edge
$(b_2,\stel{aby}{2})$, vertex $\stel{aby}{2}$ can be neither in $[A \dots a_2]$
nor in $[y_2 \dots A]$, as otherwise $(b_2,\stel{aby}{2})$ crosses $(x_2,B) \in
\BS$, $(a_2,y_2) \in \GS$ and either $(A,x_2) \in \RS$ or $(A,x_3) \in \RS$,
respectively, which is a contradiction by Lemma~\ref{lem:no3}.\ref{no3:cross3}.
Similarly, due to edge $(y_2,\stel{aby}{2})$, vertex $\stel{aby}{2}$ cannot be
in $[a_2 \dots b_2]$. Finally, due to edge $(a_2,\stel{aby}{2})$, vertex
$\stel{aby}{2}$ cannot be in $[b_2 \dots y_2]$. Hence, there is no feasible
placement of $\stel{aby}{2}$ in $\mathcal{E}$, which is a contradiction.

\item \label{p:AbxaBy} $[A \dots b_1 \dots x_1 \dots a_1 \dots b_2 \dots x_2 \dots a_2 \dots b_3 \dots x_3 \dots a_3 \dots B \dots y_3 \dots y_2 \dots y_1]$

Refer to  Fig.~\ref{fig:AbxaBy}. This case can be led to a contradiction following the reasoning of~FP\ref*{p:AxxByy}.\ref{p:AaxbBy}.

\item \label{p:AabxBy} $[A \dots a_1 \dots b_1 \dots x_1 \dots a_2 \dots b_2 \dots x_2 \dots a_3 \dots b_3 \dots x_3 \dots B \dots y_3 \dots y_2 \dots y_1]$

Refer to Fig.~\ref{fig:AabxBy}. Since edges $(A,x_3)$, $(x_1,B)$ and $(a_2,y_2)$ form a $3$-twist, we can assume that $(A,x_i) \in \RS$, $(a_i,y_i) \in \GS$ and $(x_i,B) \in \BS$, which implies that $(b_i,y_i) \in \GS$, $(a_i,B) \in \BS$, $(A,b_i) \in \RS$ and $(a_i,x_i) \in \BS$. 
Consider now vertex $\stel{Bax}{2}$ of $G_N$ that was introduced due to the stellation of face $\langle B, a_2, x_2 \rangle$ in $G_N$. 
Due to edge $(a_2,\stel{Bax}{2})$, vertex $\stel{Bax}{2}$ cannot be in $[x_2 \dots y_2]$. Analogously, vertex $\stel{Bax}{2}$ cannot be in $[y_2 \dots a_2]$, due to edge $(x_2,\stel{Bax}{2})$. Finally, vertex $\stel{Bax}{2}$ cannot be in $[a_2 \dots x_2]$, due to edge $(B,\stel{Bax}{2})$. Hence, there is no feasible placement of $\stel{Bax}{2}$ in $\mathcal{E}$; a contradiction.

\item \label{p:AbaxBy} $[A \dots b_1 \dots a_1 \dots x_1 \dots b_2 \dots a_2 \dots x_2 \dots b_3 \dots a_3 \dots x_3 \dots B \dots y_3 \dots y_2 \dots y_1]$

Refer to Fig.~\ref{fig:AbaxBy}. Since edges $(A,x_3)$, $(x_1,B)$ and $(a_2,y_2)$ form a $3$-twist, we can assume that $(A,x_i) \in \RS$, $(a_i,y_i) \in \GS$ and $(x_i,B) \in \BS$. Hence, $(a_i,B),(x_i,B) \in \BS$, $(b_i,y_i) \in \GS$ and $(A,b_i),(b_i,x_i) \in \RS$. 
It is not hard to see that there is no feasible placement for vertex $\stel{Abx}{2}$ of $G_N$ introduced due to the stellation of face $\langle A, b_2, x_2 \rangle$~in~$\mathcal{E}$.
\setcounter{AxxByy}{\value{enumi}}
\end{enumerate} 
 
\noindent We now consider the case in which the satellite vertices $a_i$ and
$b_i$ are on different sides of $(A,B)$. Since $A \prec a_1 \prec a_2 \prec a_3
\prec B$, either $B \prec b_1 \prec b_2 \prec b_3$ or $B \prec b_3 \prec b_2
\prec b_1$ holds. If $b_1 \prec b_2 \prec b_3$, then a $3$-twist is formed by
the twin edges $(b_1,x_1)$, $(b_2,x_2)$ and $(b_3,x_3)$, which is a
contradiction by Lemma~\ref{lm:twin-edges}. Hence, $b_3 \prec b_2 \prec b_1$
must hold. By the monotonic property, the partial order of vertices $A$, $B$ and
of the vertices in $\{x_i,y_i,a_i,b_i;\ i=1,2,3\}$ is one of the following:

\begin{enumerate}[\bf FP\ref*{p:AxxByy}.1]
\setcounter{enumi}{\value{AxxByy}}
\item \label{p:AxaBby} $[A \dots x_1 \dots a_1 \dots x_2 \dots a_2 \dots x_3 \dots a_3 \dots B \dots b_3 \dots y_3 \dots b_2 \dots y_2 \dots b_1 \dots y_1]$

Edges $(A,x_3)$, $(x_1,B)$, $(x_2, b_2)$, $(a_2, y_2)$ form a $4$-twist; a contradiction by Lemma~\ref{lem:no3}.\ref{no3:twist}.

\item \label{p:AaxByb} $[A \dots a_1 \dots x_1 \dots a_2 \dots x_2 \dots a_3 \dots x_3 \dots B \dots y_3 \dots b_3 \dots y_2 \dots b_2 \dots y_1 \dots b_1]$

Edges $(A,x_3)$, $(x_1,B)$, $(x_2, b_2)$, $(a_2, y_2)$ form a $4$-twist; a contradiction by Lemma~\ref{lem:no3}.\ref{no3:twist}.

\item \label{p:AaxBby} $[A \dots a_1 \dots x_1 \dots a_2 \dots x_2 \dots a_3 \dots x_3 \dots B \dots b_3 \dots y_3 \dots b_2 \dots y_2 \dots b_1 \dots y_1]$

\begin{figure}[t!]
	\centering
    \begin{subfigure}[b]{.24\textwidth}
		\centering
		\includegraphics[width=\textwidth,page=5]{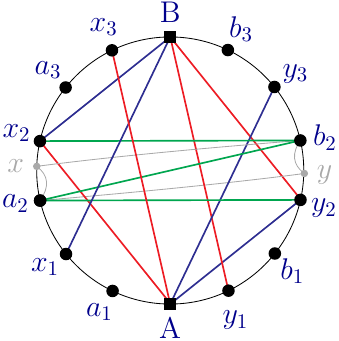}
		\caption{}
		\label{fig:AaxBby}
	\end{subfigure}
	\hfil
	\begin{subfigure}[b]{.24\textwidth}
		\centering
		\includegraphics[width=\textwidth,page=6]{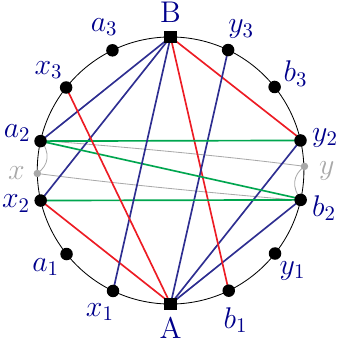}
		\caption{}
		\label{fig:AxaBby}
	\end{subfigure}
    \caption{Illustrations for 
    (a)~FP\ref*{p:AxxByy}.\ref{p:AaxBby} and (b)~FP\ref*{p:AxxByy}.\ref{p:AxaByb}.}
\end{figure}

As opposed to
FP\ref*{p:AxxByy}.\ref{p:AaxbBy}--FP\ref*{p:AxxByy}.\ref{p:AaxByb}, we do not
directly rule out this case. Instead, we identify a copy of $G_N$ in the final
graph $G$ (see Section~\ref{ssec:idea}) for which the preconditions of
Case~C.\ref{c:AxyB} hold. Thus, we reduce this case to C.\ref{c:AxyB}, for which
a direct contradiction is shown below.

Refer to Fig.~\ref{fig:AaxBby}. Since edges $(A,x_3)$, $(x_1,B)$ and $(a_2,y_2)$
form a $3$-twist, by the majority property, we may assume that $(A,x_i) \in
\RS$, $(a_i,y_i) \in \GS$ and $(x_i,B) \in \BS$. Hence, $(b_i,x_i) \in \GS$ and
$(a_i,b_i) \in \GS$. Since $(a_i,y_i) \in \GS$ and since edges $(A,y_3)$,
$(y_1,B)$ and $(a_2,y_2)$ also form a $3$-twist, by the majority property, we
may further assume that either $(A,y_i) \in \BS$ and $(y_i,B) \in \RS$, or
$(A,y_i) \in \RS$ and $(y_i,B) \in \BS$. In the following, we discuss the former
case; the latter is analogous.

Consider the copy $H_N$ of graph $G_N$ that is attached along the satellite edge
$(a_2,b_2)$ in the final graph $G$, and let $Q_{(a_2,b_2)}$ be \textit{any} copy
of graph $Q$ in~$H_N$. We  prove that no two terminals of $Q_{(a_2,b_2)}$ are on
opposite sides of $(a_2,b_2)$. Assume the contrary, which implies that there
exist two consecutive terminals, say $x$ and $y$, of $Q_{(a_2,b_2)}$ that are on
opposite sides of $(a_2,b_2)$. It is not difficult to see that either $x$ is in
$[a_2 \dots x_2]$ and~$y$ is in $[b_2 \dots y_2]$, or vice versa; see
Fig.~\ref{fig:AaxBby}. By construction of graph $Q_{(a_2,b_2)}$, vertices $x$
and~$y$ are connected by a path of length~$2$ in
$Q_{(a_2,b_2)}\setminus\{a_2,b_2\}$. Let $z$ be the intermediate vertex of this
path. Due to edge $(x,z)$, vertex $z$ can be only in $[x_1 \dots x_2]$, which
implies that its second edge $(z,y)$ crosses three edges of different colors; a
contradiction by Lemma~\ref{lem:no3}.\ref{no3:cross3}. Hence, all terminals of
$Q_{(a_2,b_2)}$ are on the same side of $(a_2,b_2)$. As this property holds for
all copies of graph $Q$ in $H_N$, Case~C.\ref{c:AxyB} applies for graph $H_N$,
as we mentioned above.

\item \label{p:AxaByb} $[A \dots x_1 \dots a_1 \dots x_2 \dots a_2 \dots x_3 \dots a_3 \dots B \dots y_3 \dots b_3 \dots y_2 \dots b_2 \dots y_1 \dots b_1]$

Refer to  Fig.~\ref{fig:AxaBby}. This case can be reduced to C.\ref{c:AxyB} closely following the reasoning of~FP\ref*{p:AxxByy}.\ref{p:AaxBby}.
\end{enumerate}
This concludes the discussion of Case C.\ref{c:AxBy} in which there exist two
terminals of $\mathcal{Q}_1$ (and thus, of
$\mathcal{Q}_2,\ldots,\mathcal{Q}_\kappa$) that are on opposite sides of edge
$(A,B)$ in $\mathcal{E}$.
\end{proof}

\medskip\noindent\textbf{Case~C.\ref*{c:AxyB}}: We next rule out the case in
which all terminals of $\mathcal{Q}_1$ (and, thus of
$\mathcal{Q}_2,\ldots,\mathcal{Q}_\kappa$) are on the same side of $(A,B)$ in
$\mathcal{E}$. By Fact~\ref{fact:2} applied on $\mathcal{Q}_1$, we may assume
that there exist two terminals, and thus two consecutive terminals $\langle x_1,
y_1 \rangle$, of $\mathcal{Q}_1$ such that either edges $(A,x_1)$ and $(A,y_1)$,
or edges $(B,x_1)$ and $(B,y_1)$ have been assigned to different pages in
$\mathcal{E}$. Assume w.l.o.g.\ that $(B,x_1) \in \RS$ and $(B,y_1) \in \GS$.
Since $G_N$ is symmetric with respect to $A$ and $B$, we may further assume
w.l.o.g.\ that $A \prec x_1 \prec y_1 \prec B$. By the majority property, the
corresponding terminals $\langle x_2, y_2 \rangle, \ldots, \langle x_\kappa,
y_\kappa \rangle$ of $\mathcal{Q}_2,\ldots,\mathcal{Q}_\kappa$ are also between
$A$ and $B$ in $\prec$, and $(B,x_i) \in \RS$ and $(B,y_i) \in \GS$, for each
$i=1,\ldots,\kappa$. W.l.o.g., let $x_1 \prec \ldots \prec x_\kappa$. Finally,
let $\langle a_1, b_1 \rangle, \ldots, \langle a_\kappa, b_\kappa  \rangle$ be
the corresponding satellite vertices of $\langle x_1, y_1 \rangle, \ldots,
\langle x_\kappa, y_\kappa  \rangle$. %
By Lemma~\ref{lem:pair-conf}, there are three subcases to consider, namely, the
pairs $\langle x_1,y_1 \rangle, \ldots, \langle x_\kappa, y_\kappa \rangle$ can
form a $\kappa$-twist, a $\kappa$-rainbow, or a $\kappa$-necklace; refer to
Forbidden Patterns~\ref{p:AxxfyyB}, \ref{p:AxxyyB} and~\ref{p:AxyxyB},
respectively. To rule out the first two, it suffices to assume $\kappa = 3$.
However, for the last one we use a larger value for $\kappa$.

\begin{pattern}\label{p:AxxfyyB}
$[A \ldots x_1 \ldots x_2 \ldots x_3 \ldots y_1 \ldots y_2 \ldots y_3 \ldots B]$
\end{pattern}
\begin{proof}
Let $i \in \{1,2,3\}$. Since edge $(A,y_2)$ crosses both $(B,x_1) \in \RS$ and $(B,y_1) \in \GS$, by the majority property that $(A,y_i) \in \BS$. Similarly, $(A,x_i) \in \BS$ or $(A,x_i) \in \GS$. 
\begin{adjustwidth}{1.5em}{}
\begin{proposition}\label{prop:AabB}
$x_1 \prec a_i \prec y_3$ and $x_1 \prec  b_i \prec y_3$
\end{proposition}
\begin{proof}
Assume to the contrary that $a_i \prec x_1$ or $y_3 \prec a_i$. If $a_i \prec
x_1$ or $B \prec a_i$, edge $(a_2,y_2)$ crosses $(A,y_1) \in \BS$, $(B,x_1) \in
\RS$, $(B,y_1) \in \GS$, a contradiction by
Lemma~\ref{lem:no3}.\ref{no3:cross3}; see Fig.~\ref{fig:AabB}. Otherwise ($y_3
\prec a_i \prec B$), edge $(a_2,x_2)$ crosses $(A,y_1) \in \BS$, $(B,x_3) \in
\RS$, $(B,y_1) \in \GS$, a contradiction by
Lemma~\ref{lem:no3}.\ref{no3:cross3}. The proof of the other claim is analogous.
\end{proof}

\begin{figure}[t!]
	\centering
    \begin{subfigure}[b]{.24\textwidth}
		\centering
		\includegraphics[width=\textwidth,page=2]{pics/AxyB}
		\caption{}
		\label{fig:AabB}
	\end{subfigure}
	\hfil
    \begin{subfigure}[b]{.24\textwidth}
		\centering
		\includegraphics[width=\textwidth,page=3]{pics/AxyB}
		\caption{}
		\label{fig:xiaibi}
	\end{subfigure}
	\hfil
    \begin{subfigure}[b]{.24\textwidth}
		\centering
		\includegraphics[width=\textwidth,page=4]{pics/AxyB}
		\caption{}
		\label{fig:axblue}
	\end{subfigure}
	\hfil
    \begin{subfigure}[b]{.24\textwidth}
		\centering
		\includegraphics[width=\textwidth,page=1]{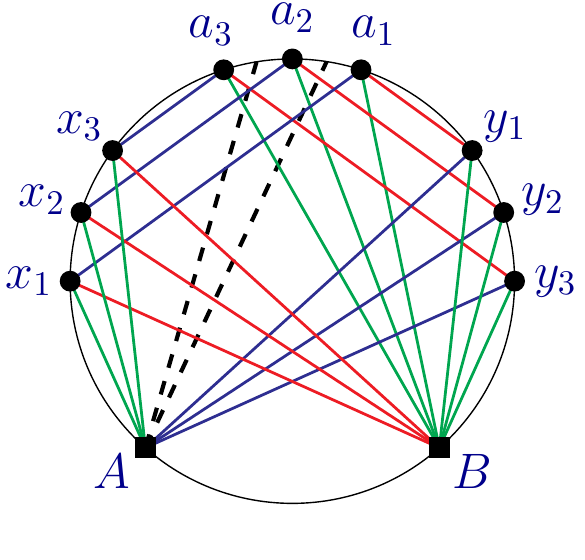}
		\caption{}
		\label{fig:axgreen}
	\end{subfigure}
   \caption{%
   Illustrations for 
   Forbidden Pattern~\ref{p:AxxfyyB}.}
    \label{fig:AxxfyyB}
\end{figure}

\begin{proposition}\label{prop:x3aiy1}
$x_3 \prec a_3 \prec a_2 \prec a_1 \prec y_1$ and $x_3 \prec b_3 \prec b_2 \prec b_1 \prec y_1$
\end{proposition}
\begin{proof}
We argue for the former; the latter is analogous. If the twin edges $(a_1,x_1)$,
$(a_2,x_2)$ and $(a_3,x_3)$ form a $3$-necklace, then $a_2$ is in $[x_1 \dots
x_3]$, which implies that edge $(a_2,y_2)$ crosses $(A,y_1) \in \BS$,
$(B,x_3)\in \RS$ and $(B,y_1)\in \GS$ (see Fig.~\ref{fig:xiaibi}); a
contradiction by Lemma~\ref{lem:no3}.\ref{no3:cross3}. Hence by
Lemma~\ref{lm:twin-edges}, $(a_1,x_1)$, $(a_2,x_2)$ and $(a_3,x_3)$ form a
$3$-rainbow. Similarly, we argue that edges $(a_1,y_1)$, $(a_2,y_2)$ and
$(a_3,y_3)$ also form a $3$-rainbow. Since the two $3$-rainbows share $a_1$,
$a_2$ and $a_3$, the proof follows from Proposition~\ref{prop:AabB}.
\end{proof}
\end{adjustwidth}

\noindent If $(A,x_i) \in \BS$, then $(x_i, a_i) \in \GS$ and $(B,a_i) \in \RS$,
which imply that edge $(a_3,y_3)$ crosses $(A, y_1) \in \BS$, $(x_1, a_1) \in
\GS$ and $(B,a_2) \in \RS$ (see Fig.~\ref{fig:axblue}); a contradiction by
Lemma~\ref{lem:no3}.\ref{no3:cross3}. Hence, $(A,x_i) \in \GS$. This implies
that $(x_2,a_2) \in \BS$ and $(a_2,y_2) \in \RS$; see Fig.~\ref{fig:axgreen}.
Since by majority property $(x_i,a_i) \in \BS$ and $(a_i,y_i) \in \RS$, it
follows that $(B,a_i) \in \GS$. We next argue about $b_2$. By
Proposition~\ref{prop:x3aiy1}, $b_2$ is in $[x_3 \dots y_1]$. In the presence of
$a_1$, $a_2$ and $a_3$ in the same interval, we can further restrict the
placement of $b_2$ either in $[a_3 \dots a_2]$ or in $[a_2 \dots a_1]$. However,
in both cases edge $(A,b_2)$ crosses three edges of different colors, namely,
$(B,a_3) \in \GS$, $(x_1,a_1) \in \BS$ and $(a_3,y_3) \in \RS$, which is a
contradiction by Lemma~\ref{lem:no3}.\ref{no3:cross3}.
\end{proof}

\begin{pattern}\label{p:AxxyyB}
$[A \ldots x_1 \ldots x_2 \ldots x_3 \ldots y_3 \ldots y_2 \ldots y_1 \ldots B]$
\end{pattern}
\begin{proof}
Since $(B,x_i) \in \RS$ and $(B,y_i) \in \GS$, as in the proof of Forbidden
Pattern~\ref{p:AxxfyyB}, we prove that $(A,y_i) \in \BS$, and that either 
$(A,x_i) \in \BS$ or $(A,x_i) \in \GS$. As in the proof of
Proposition~\ref{prop:x3aiy1}, we can further prove that a $3$-rainbow is formed
both by the twin edges $(a_1,x_1)$, $(a_2,x_2)$ and $(a_3,x_3)$ and by the twin
edges $(a_1,y_1)$, $(a_2,y_2)$ and $(a_3,y_3)$. Since both rainbows share $a_1$,
$a_2$ and $a_3$, it is not possible that they exist simultaneously due to the
underlying order $[x_1 \ldots x_2 \ldots x_3 \ldots y_3 \ldots y_2 \ldots y_1]$.
\end{proof}

\begin{pattern}\label{p:AxyxyB}
$[A \dots x_1 \dots y_1 \dots x_2 \dots y_2 \;\; \dots \;\; x_\kappa \dots y_\kappa \dots B]$
\end{pattern}
\begin{proof}
Let $i \in \{1,2 \ldots, \kappa\}$. Recall that $(B, x_i) \in \RS$ and $(B, y_i)
\in \GS$. Since each of $(A,x_2)$ and $(A,y_2)$ crosses both $(B,x_1) \in \RS$
and $(B,y_1) \in \GS$, by majority property it follows that $(A, x_i) \in \BS$
and $(A, y_i) \in \BS$; see Fig.~\ref{fig:AxyxyB}. To rule out this case, we
assume that $\kappa$ is even, such that $\kappa > d_Q + 4$, where $d_Q$ denotes
the length of the maximum shortest path between a terminal of graph $Q$ and
every other vertex of it that passes neither through $A$ nor$B$. Note that $d_Q
\neq n_Q$. Consider the copy $\mathcal{Q}_{\kappa/2}$ of graph $Q$, to which the
terminals $x_{\kappa/2}$ and $y_{\kappa/2}$ belong. By Case C.\ref{c:AxyB}, all
terminals of $\mathcal{Q}_{\kappa/2}$ are in $[A \dots B]$ in $\mathcal{E}$. In
the following, we show that all the vertices of $\mathcal{Q}_{\kappa/2}$ that
are different from $A$ and $B$ are in $[y_1 \dots x_{\kappa}]$. This implies
that each of the terminals of $\mathcal{Q}_{\kappa/2}$ is connected to $A$
through an edge of the $\BS$ page (as it is involved in crossings with  $(B,x_1)
\in \RS$ and $(B,y_1) \in \GS$), and to $B$ through an edge of either the $\RS$
or of the $\GS$ page (as it is involved in a crossing with $(A,y_1) \in \BS$).
The contradiction is obtained 
by Fact~\ref{fact:1} applied to $\mathcal{Q}_{\kappa/2}$, whose
preconditions~(\ref{fact1:i})--(\ref{fact1:iii}) are met as discussed above.
	
\begin{figure}
	\centering
	\includegraphics[page=5,scale=0.75]{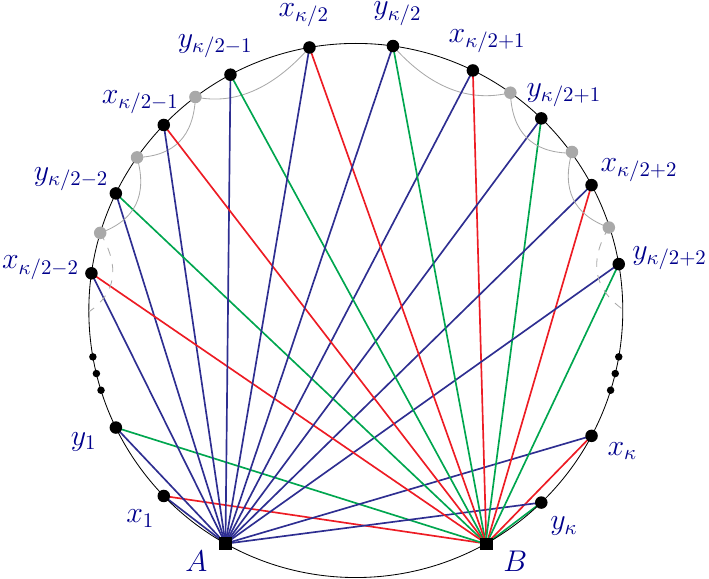}
	\caption{Illustration for Forbidden Pattern~\ref{p:AxyxyB}.}
	\label{fig:AxyxyB}
\end{figure}
	
To complete the proof, we observe that all the vertices of
$\mathcal{Q}_{\kappa/2}$ that are different from $A$ and $B$ and at distance~$1$
either from $x_{\kappa/2}$ or from $y_{\kappa/2}$ lie in $[x_{\kappa/2-1} \dots
y_{\kappa/2+1}]$, as otherwise an edge incident to $x_{\kappa/2}$ or 
$y_{\kappa/2}$ is inevitably crossing three edges of different colors; a
contradiction by Lemma~\ref{lem:no3}.\ref{no3:cross3}. By induction, we obtain
that all the vertices that are different from $A$ and $B$ and at distance~$j$
either from $x_{\kappa/2}$ or from $y_{\kappa/2}$ lie in $[y_{\kappa/2- j/2 -1}
\dots x_{\kappa/2+j/2 +1}]$, if $j$ is even, and in $[x_{\kappa/2-\lfloor j/2
	\rfloor-1} \dots y_{\kappa/2+\lfloor j/2 \rfloor+1}]$, if $j$ is odd. By the
definition of $d_Q$, any vertex of $\mathcal{Q}_{\kappa/2}$ different from $A$
and $B$ is in $[x_{\kappa/2-\lfloor d_Q/2 \rfloor - 1} \dots y_{\kappa/2+\lfloor
	d_Q/2 \rfloor + 1}]$, which by the choice of $\kappa$ is in $[x_2 \dots
y_{\kappa-1}]$, and thus in $[y_1 \dots x_k]$, as desired.
\end{proof}

\noindent By Cases~C.\ref{c:AxBy} and C.\ref{c:AxyB}, it follows that graph $G$
does not admit a $3$-page book embedding, which completes the proof of
Theorem~\ref{thm:main}.

We conclude this section with some insights on the size of graph $G$. For most
of the patterns that we proved to be forbidden, the value of $\kappa$ is~$3$.
However, in Forbidden Pattern~\ref{p:AxyxyB}, this value is increased to
$d_Q+5$, which equals $28$. Using this value, one can compute the number $N$ of
copies of graph $Q$ in the base graph $G_N$ with $n_Q = 354$, $m_Q = 1,\!056$
and $p_Q=62,\!481$. Since each of the $N$ copies of graph $Q$ in the base graph
$G_N$ gives rise to nine copies of the base graph in the final graph $G$, the
size of graph $G$ is enormously large. In the next section, we present a
considerably smaller graph that serves as a certificate to
Theorem~\ref{thm:main}.

\section{A Computer-Aided Proof}
\label{sec:sat}

In this section, we first briefly recall an efficient automatic approach for
computing book embeddings with certain number of pages that was first proposed
in \cite{DBLP:conf/gd/Bekos0Z15}. Then, we apply this approach (with appropriate
modifications) to find a medium-sized planar graph that requires four pages, and
to verify Facts~\ref{fact:1}~and~\ref{fact:2} for $Q_7$ and $Q_{10}$,
respectively.

To formulate the book embedding problem as a SAT instance,  Bekos et
al.~\cite{DBLP:conf/gd/Bekos0Z15} use three different types of variables,
denoted by $\sigma$, $\phi$ and $\chi$, with the following meanings:
\begin{inparaenum}[(i)]
\item for a pair of vertices $u$ and $v$, variable $\sigma(u,v)$ is $\texttt{true}$, if and only if $u$ is to the left of $v$ along
the spine, 
\item for an edge $e$ and a page $\rho$, variable $\phi_\rho(e)$ is $\texttt{true}$, if and only if edge $e$ is assigned to page $\rho$ of the book, and 
\item for a pair of edges $e$ and $e'$, variable $\chi(e,e')$ is $\texttt{true}$, if and only if $e$ and $e'$ are assigned to the same page.
\end{inparaenum}  
Hence, there exist in total $O(n^2+m^2+pm)$ variables, where $n$ denotes the number of vertices of the graph, $m$ its number of edges, and $p$ the number of available pages. A set of $O(n^3+m^2)$ clauses ensures that the underlying order is indeed linear, and that no two edges of the same page cross; for details we point the reader to~\cite{DBLP:conf/gd/Bekos0Z15}.

Using the above SAT formulation, we are able to test various planar graphs on
$3$-page embeddability. One that does not admit a $3$-page book embedding (see
Fig.~\ref{fig:needs4pages}) is constructed from graph $Q_8$ by removing the edge
connecting its poles $A$ and $B$ and by identifying its opposite terminals,
$t_0$ and $t_7$. %
Formally, start with an embedded $Q_8=(V_8, E_8)$ having the outerface $\langle
A, t_0, B, t_7 \rangle$ after the removal of $(A,B)$. Then, \df{contract} $t_0$
and $t_7$, that is, create a new graph $Q_8^{\circ}=(V_8^{\circ}, E_8^{\circ})$
in which
\begin{inparaenum}[(i)]
\item $V_8^{\circ} = V_8 \setminus \{t_8\}$, 
\item for every edge $(u, v) \in E_8$ such that $u \neq t_7, v \neq t_7$, there exists a corresponding edge $(u, v) \in E_8^{\circ}$, and
\item for every edge $(v, t_7) \in E_8$, there exists a corresponding edge $(v, t_0) \in E_8^{\circ}$; refer to Fig.~\ref{fig:needs4pages} for an illustration.
\end{inparaenum} 
It is easy to see that the contraction of $t_0$ and $t_8$ can be done in a planarity-preserving way, and hence, $Q_8^{\circ}$ is maximal planar with $275$ vertices and $819$ edges.  Observe that the graph has treewidth $4$; this is in contrast with planar graphs of treewidth $3$ that always admit $3$-page book embeddings~\cite{DBLP:conf/focs/Heath84}. 

\begin{figure}[t!]
	\centering
	\includegraphics[width=.9\textwidth,page=2]{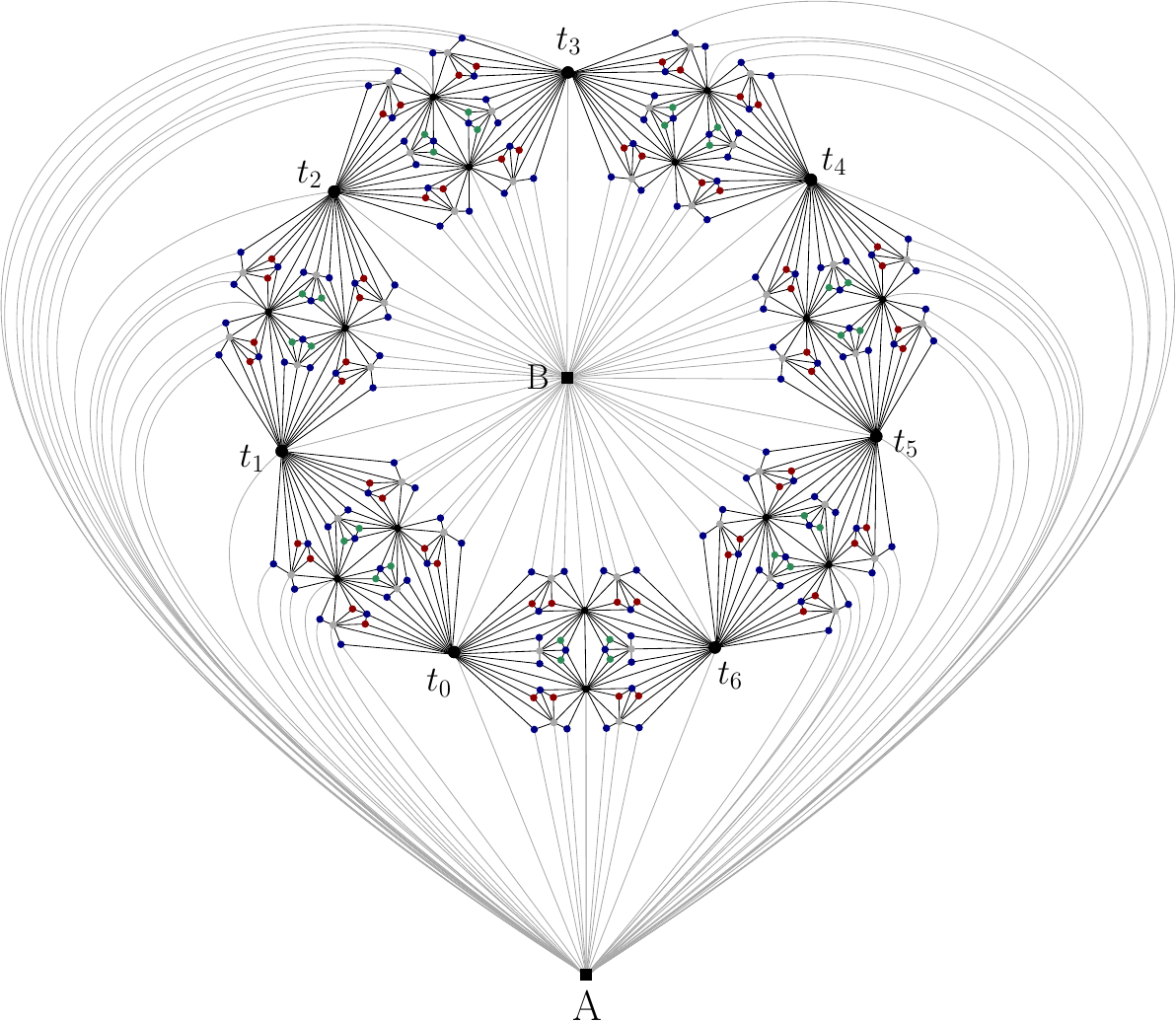}
	\caption{Illustration of graph $Q_8^\circ$ consisting of $275$ vertices and $819$ edges.}
	\label{fig:needs4pages}
\end{figure}

Our early attempts to verify $3$-page embeddability of $Q_8^{\circ}$ were
unsuccessful due to an enormous search space of possible satisfying assignments.
To reduce the search space, we introduce several \df{symmetry-breaking}
constraints, that is, variable assignments that preserve the satisfiability of
an instance:
\begin{itemize}
\item[--] we choose a particular vertex as the first one along the spine: $\sigma(A, v)$ for every $v \in V_8^{\circ} \setminus \{A\}$;

\item[--] since $Q_8^{\circ}$ is symmetric with respect to the terminals, we select $t_0$ to be the first among the terminals in the vertex ordering: $\sigma(t_0, t_i)$ for all $0 < i \le 6$;

\item[--] a vertex ordering can be reversed without affecting its book embeddability; we introduce a rule so that the SAT instance contains only one of the two possible solutions: $\sigma(t_1, t_2)$;

\item[--] to break symmetries of page assignments, we fix an edge to a particular page: $\phi_1(A, t_0)$;

\item[--] similarly, another edge can be assigned to one of the first two pages: $\phi_1(B, t_0) \vee \phi_2(B, t_0)$;

\item[--] since $K_4$ is not $1$-page book embeddable (as it is not outerplanar), we impose for every $K_4$ subgraph of $Q_8^{\circ}$ that not all its edges are assigned to the same page, namely, for every such a subgraph with  edges $e_1,\ldots,e_6$ we set: $\neg \phi_\rho(e_1) \vee \ldots \vee \neg \phi_\rho(e_6), \; \forall \;  1 \le \rho \le 3$.
\end{itemize}	

With two independent implementations of~\cite{DBLP:conf/gd/Bekos0Z15} and using
the above extra rules, we are able to verify that graph $Q_8^{\circ}$ is not
$3$-page embeddable, thus providing an alternative proof to
Theorem~\ref{thm:main}. The source codes of both implementations are available
to the community at~\cite{alice,bob}.

The first implementation~\cite{bob} was executed on a dual-node 28-core 2.4 GHz
Intel Xeon E5-2680 machine with $256$GB RAM. To verify unsatisfiability, we used
the \texttt{plingeling}~\cite{lingeling} parallel SAT solver, which needed
approx.\ $48$ hours using $56$ available threads. %
The second implementation~\cite{alice} was executed on a much weaker single-node
$4$-core 3.3 GHz Intel Core~i5-4590 machine with $16$GB RAM. Since the machine
is weaker, to verify unsatisfiability, we split the actual problem into
subproblems depending on the number of terminals between $A$ and $B$. Since the
graph is symmetric with respect to $A$ and $B$, it is enough to assume that
there exist $0$, $1$, $2$ or $3$ terminals between $A$ and $B$. For each of the
cases, we further distinguish subcases depending on the relative order of these
terminals. In total, we consider $28$ subproblems, which we solved using the
\texttt{lingeling}~\cite{lingeling} SAT solver on a single thread. The total
time needed to verify unsatisfiability of these subproblems was approx.\ $35$
hours.

We emphasize that $Q_8^{\circ}$ is likely the minimal graph from the considered
family that requires four pages. For example, an analogously constructed
$Q_7^{\circ}$, as well as non-contracted variants, $Q_k$, with up to $k = 10$,
do admit  book embeddings in three pages. Similarly, performing fewer
stellations yields $3$-page embeddable instances.

Unlike computationally expensive processing of $Q_8^{\circ}$, our approach is
very efficient for verification of Fact~\ref{fact:1} and Fact~\ref{fact:2}. The
main reason is that the generated SAT instances contain more constraints, which
significantly reduce the search space of possible solutions. We use the same two
implementations to verify that $Q_7$ does not admit a 3-page book embedding
under the restrictions of Fact~\ref{fact:1} and that $Q_{10}$ does not admit a
3-page book embedding under the restrictions of Fact~\ref{fact:2}. Both
implementations are able to process the graphs within several minutes, even when
a single-threaded SAT solver is utilized. Again we stress that the two graphs
are minimal in the considered family that satisfy the properties of the facts.

\section{Conclusion}
\label{sec:conclusions}

By closing the gap between the lower bound and the upper bound
on the book thickness of planar graphs, we resolved a problem that remained open for more than thirty years. We mention three interesting research directions that are related to our work.

\begin{enumerate}
\item There exist several subclasses of planar graphs with book thickness two
proposed in the literature. For example, $4$-connected planar
graphs~\cite{NC08}, planar graphs without separating
triangles~\cite{DBLP:journals/appml/KainenO07}, Halin
graphs~\cite{DBLP:journals/mp/CornuejolsNP83}, series-parallel
graphs~\cite{DBLP:conf/cocoon/RengarajanM95}, bipartite planar
graphs~\cite{DBLP:journals/dcg/FraysseixMP95}, planar graphs of maximum
degree~4~\cite{DBLP:journals/algorithmica/BekosGR16}, triconnected planar graphs
of maximum degree~5~\cite{DBLP:conf/esa/0001K19}, and maximal planar graphs of
maximum degree~6~\cite{Ewald1973}. On the other hand, the planar graphs with
book thickness three are less studied, and to the best of our knowledge only
include the class of planar $3$-trees~\cite{DBLP:conf/focs/Heath84}. Recently,
Guan and Yang~\cite{DBLP:journals/dam/GuanY2019} suggested an algorithm to embed
general (that is, not necessarily triconnected) planar graphs of maximum
degree~$5$ in books with three pages, but it is not known whether there exist
such graphs that require three pages (an open problem of independent research
interest). Here, we suggest to study other natural subclasses of planar graphs
with book thickness three. Two candidates are:
\begin{inparaenum}
\item the class of planar Laman graphs, and
\item the class of planar graphs with bounded maximum degree $\Delta \geq 7$. 
\end{inparaenum}
Note that both classes contain members that are not $2$-page book embeddable.

\item In the literature, book embeddings are also known as \df{stack layouts},
since the edges assigned to the same page (called \df{stack} in this context)
follow the last-in-first-out model in the underlying linear order. The ``dual''
concept of a book embedding is the so-called \df{queue layout} in which the
edges assigned to the same page (called \df{queue} in this context) follow the
first-in-first-out model. A recent breakthrough result by Dujmovi{\'{c}} et
al.~\cite{DBLP:conf/focs/DujmovicJMMUW19} suggests that planar graphs admit
queue layouts with at most 49 queues. Here, we are asking whether planar graphs
admit \emph{mixed} layouts with $s$ stacks and $q$ queues for some $s < 4$ and
$q < 49$? Such mixed layouts partition the edges of a graph into $s$ stacks and
$q$ queues, while using a common vertex ordering; they have been introduced by
Heath, Leighton and Rosenberg~\cite{DBLP:journals/siamdm/HeathLR92}.
Pupyrev~\cite{DBLP:conf/gd/Pupyrev17} showed that one stack and one queue do not
suffice for planar graphs, while de Col et al.~\cite{DBLP:conf/gd/ColKN19}
proved that testing the existence of a $2$-stack $1$-queue layout of general
(non-planar) graphs is \NP-complete.
	
\item Finally, we would like to see progress on the book thickness of planar
directed acyclic graphs (DAGs). Note that in the directed version of the book
embedding problem, the edge directions must be consistent with the constructed
vertex ordering. Heath et
al.~\cite{DBLP:journals/siamcomp/HeathPT99a,DBLP:journals/siamcomp/HeathP99b}
asked whether the book thickness of upward planar DAG is bounded by a constant,
and they provided constant bounds for directed trees, unicyclic DAGs, and
series-parallel DAGs. Frati et al.~\cite{DBLP:journals/jgaa/FratiFR13} extended
their results in the upward planar triangulations of bounded diameter or of
bounded maximum degree. However, the general question remains open.
\end{enumerate}	

\medskip\noindent\textit{Acknowledgment.} This work started at the GNV'19 workshop (30 June - 5 July, 2019, Heiligkreuztal, Germany). We thank the participants for fruitful discussions.

\bibliographystyle{splncs03}
\bibliography{stacks,general}

\end{document}